\documentclass[]{jac}

\usepackage[utf8]{inputenc}
\usepackage[T1]{fontenc}
\usepackage{amsmath,amssymb,latexsym}
\usepackage{color}

\usepackage{ifpdf}
\ifpdf
\usepackage[pdftex]{graphicx}
\else
\usepackage{graphicx}
\fi

\theoremstyle{plain}\newtheorem{claim}[thm]{Claim}

\newcommand{\NN}{\mathbb{N}}
\newcommand{\ZZ}{\mathbb{Z}}

\newcommand{\Lang}{\mathcal{L}}

\newcommand{\A}{\mathcal{A}}
\newcommand{\Si}{\Sigma}

\title{Construction of $\mu$-limit Sets}
\author[LAMA]{L. Boyer}{Laurent Boyer}
\author[LIF]{M. Delacourt}{Martin Delacourt}
\author[LATP]{M. Sablik}{Mathieu Sablik}
\address[LAMA]{LAMA, Université de Savoie}
\address[LIF]{Laboratoire d'Informatique Fondamentale de Marseille, Université de Provence}
\address[LATP]{Laboratoire d'Analyse, Topologie, Probabilités, Université de Provence}
\thanks{Thanks to the  project ANR EMC: ANR-09-BLAN-0164.}
\begin{document}

\ifpdf
\DeclareGraphicsExtensions{.pdf, .jpg, .tif}
\else
\DeclareGraphicsExtensions{.eps, .ps, .jpg}
\fi

\maketitle

\begin{abstract}\noindent
The $\mu$-limit set of a cellular automaton is a subshift whose forbidden patterns are exactly those, whose probabilities tend to zero as time tends to infinity. In this article, for a given subshift in a large class of subshifts, we propose the construction of a cellular automaton which realizes this subshift as $\mu$-limit set where $\mu$ is the uniform Bernoulli measure.
\end{abstract}

\section{Introduction} 

A cellular automaton (CA) is a complex system defined by a local rule which acts synchronously and uniformly on the configuration space. These simple models have a wide variety of different dynamical behaviors. More particularly it is interesting to understand its behavior when it goes to infinity. 

In the dynamical systems context, it is natural to study the limit set of a cellular automaton, it is defined as the set of configurations that can appear arbitrarily far in time. This set captures the longterm behavior of the CA and has been widely studied since the end of the 1980s. Given a cellular automaton, it is difficult to determine its limit set. Indeed it is undecidable to know if it contains only one configuration~\cite{Kari-1992} and more generally, every nontrivial property of limit sets is undecidable~\cite{Kari-1994}. Another problem is to characterize which subshift can be obtained as limit set of a cellular automaton. This was first studied in detail by Lyman Hurd~\cite{Hurd-1987}, and significant progress have been made~\cite{Maass-1995,Formenti-Kurka-2007} but there is still no characterization. The notion of limit set can be refined if we consider the notion of attractor~\cite{Hurley-1990-1,Kur}. 

However, these topological notions do not correspond to the empirical point of view where the initial configuration is chosen randomly, that is to say chosen according a measure $\mu$. That's why the notion of $\mu$-attractor is introduced by~\cite{Hurley-1990-2}. Like it is discussed in~\cite{KurMaa} with a lot of examples, this notion is not satisfactory empirically and the authors introduce the notion of $\mu$-limit set. A $\mu$-limit set is a subshift whose forbidden patterns are exactly those, whose probabilities tend to zero as time tends to infinity. This set corresponds to the configurations which are observed when a random configuration is iterated.

As for limit sets, it is difficult to determine the $\mu$-limit set of a given cellular automaton, indeed it is already undecidable to know if it contains only one configuration~\cite{Boyer-Poupet-Theyssier-2006}. However, in the literature, all $\mu$-limit sets which can be found are very simple (transitive subshifts of finite type). In this article, for every recursively enumerable family $(\Si_i)_{i\in\NN}$ of subshifts generated by a generic configuration, we construct a cellular automaton which realizes $\overline{\bigcup_{i\in\NN}\Si_i}$ as $\mu$-limit set. In particular all transitive sofic subshifts can be realized. It makes a strong difference with limit sets since there are sofic subshifts, as the even subshift (subshift on alphabet $\{0, 1\}$ in which all words $01^k0$ with odd $k$ are forbidden), which cannot be realized as limit set~\cite{Maass-1995}.

To construct a cellular automaton that realizes a given subshift as $\mu$-limit set, we first erase nearly all the information contained in a random configuration thanks to counters (section~\ref{sec:counters}). Then we produce segments, which are finite areas of computation. On each segment we construct small parts of the generic configurations of many subshifts, and as time passes, segments grow larger and every word of every subshift appears often enough (section \ref{sec:infinite_unions}).

\section{Definitions}

\subsection{Words and density}
For a finite set $Q$ called an \emph{alphabet}, denote $Q^{\ast}= \bigcup_{n\in\NN} Q^n$ the set of all finite words over $Q$. The \emph{length} of $u = u_0u_1\dots u_{n-1}$ is $|u| = n$. 
 We denote $Q^{\ZZ}$ the set of \emph{configurations} over $Q$, which are mappings from $\ZZ$ to $Q$, and  for $c\in Q^{\ZZ}$, we denote $c_z$ the image of $z\in \ZZ$ by $c$. For $u\in Q^{\ast}$ and $0\leq i\leq j\leq|u|-1$ we define the \emph{subword} $u_{[i,j]} =u_iu_{i+1}\dots u_j$; this definition can be extended to a configuration $c\in Q^{\ZZ}$ as $c_{[i,j]} =c_ic_{i+1}\dots c_j$ for $i,j\in\ZZ$ with $i\leq j$. The \emph{language} of $S\subset Q^{\ZZ}$ is defined by $$\Lang(S)=\{ u \in Q^{\ast} : \exists c\in S , \;\exists i \in\ZZ \textrm{ such that } u = c_{[ i , i + | u | - 1 ] }\}.$$ 

For every $u \in Q^{\ast}$ and $i\in\ZZ$, we define the \emph{cylinder} $[u]_i$ as the set of configurations containing the word $u$ in position $i$ that is to say $[u]_i = \{c\in Q^{\ZZ} : c_{[i,i+|u|-1]} = u\}$. If the cylinder is at the position $0$, we just denote it by $[u]$.

For all $u,v\in Q^{\ast}$ define $|u|_v$  the \emph{number of occurences} of $v$ in $u$ as:
$$|u|_v=\textrm{card}\{i\in[0,|u|-|v|] : u_{[i,i+|v|-1]}=v\}$$

For any two words $u,v \in Q^{\ast}$, let $d_u(v)=\frac{|u|_v}{|u|-|v|}$.\\

For a configuration $c\in Q^{\ZZ}$, the \emph{density} $d_c(v)$ of a finite word $v$ is:
\begin{displaymath}
d_c(v)=\limsup_{n\to +\infty} d_{c_{[-n,n]}}(v).
\end{displaymath}

These definitions could be generalized, for a set of words $W\subset Q^{\ast}$, we note $|u|_W$ and $d_c(W)$.

\begin{definition}[Normal configuration]
A configuration is said to be \emph{normal} for an alphabet $Q$ if all words of length $n$ have the same density of apparition in the configuration. 
\end{definition}

\subsection{Subshifts}

We denote by $\sigma$ the \emph{shift} map $\sigma:Q^{\ZZ}\mapsto Q^{\ZZ}$ defined by $\sigma(c)_i=c_{i-1}$. A \emph{subshift} is a closed, $\sigma$-invariant subset of $Q^{\ZZ}$. It is well known that a subshift is completely described by its language denoted $\Lang(\Si)$. Moreover, it is possible to define a subshift by a set of its forbidden words which do not appear in the language.

As the shift invariance is preserved, intersections and closures of unions of subshifts are still subshifts. And in particular, the union of a set $(\Lang(\Si_i))_i$ of languages describes the subshift that is the closure of the union of all subshifts: $\overline{\bigcup_{i\in\NN}\Si_i}$.

We define some classes of subshifts. A \emph{sofic subshift} is a subshift whose language of forbidden words is rational, i.e. given by a finite automaton. A  subshift $\Si$ is \emph{transitive} if for all $u,v\in\Lang(\Si)$ there exists a word $w$ such that $uwv\in\Lang(\Si)$. Let $s:Q\to Q^{\ast}$ be a primitive substitution (there exists $k\in\NN$ such that for all $a,b\in Q$ $a$ appears in $s^k(b)$), the \emph{substitutive subshift} associated to s is the subshift $\Si_s$ such that $$\Lang(\Si_s)=\{u\in Q^{\ast} :\exists \textrm{$a\in Q$ and $n\in\NN$ such that $u$ appears in $s^n(a)$}\}.$$

\subsection{Cellular automata}

\begin{definition} [Cellular automaton]
 A \emph{cellular automaton (CA)} $\A$ is a triple $(Q_{\A},r_{\A},\delta_{\A})$ where $Q_{\A}$ is a finite set of states called the \emph{alphabet}, $r_{\A}$ is the \emph{radius} of the automaton, and $\delta_{\A}:Q_{\A}^{2r_{\A}+1}\mapsto Q_{\A}$ is the \emph{local rule}.\\
 
 The configurations of a cellular automaton are the configurations over $Q_{\A}$.
 A global behavior is induced and we'll note $\A(c)$ the image of a configuration $c$ given by: $\forall z\in \ZZ, \A(c)_z=\delta_{\A}(c_{z-r_{\A}},\dots,c_z,\dots,c_{z+r_{\A}})$. Studying the dynamic of $\A$ is studying the iterations of a configuration by the map $\A:Q_{\A}^{\ZZ}\to Q_{\A}^{\ZZ}$. When there is no ambiguity, we'll note $Q$, $r$ and $\delta$ instead of $Q_{\A}$, $r_{\A}$ and $\delta_{\A}$.

\end{definition}


\subsection{$\mu$-limit sets}\label{sec:mulimites}

\begin{definition}[Uniform Bernoulli measure]
For an alphabet $Q$, the \emph{uniform Bernoulli measure} $\mu$ on configurations over $Q$ is defined by: $\forall u\in Q^*, i\in \ZZ, \mu([u]_i)=\frac{1}{|Q|^{|u|}}$.
\end{definition}

For a CA $\A=(Q,r,\delta)$ and $u\in Q^*$, we denote for all $n\in \NN$, $\A^n\mu([u])=\mu\left(\A^{-n}([u])\right)$.

\begin{definition}[Persistent set]
For a CA $\A$, and the uniform Bernoulli measure $\mu$, we define the \emph{persistent set} $L_{\mu}(\A)$ with: $\forall u \in Q^*$:
 $$u\notin L_{\mu}(\A) \Longleftrightarrow \lim_{n\rightarrow\infty}\A^n\mu([u]_0)=0.$$
 
 Then the \emph{$\mu$-limit set} of $\A$ is $\Lambda_{\mu}(\A)=\left\{c\in Q^{\ZZ}:L(c)\subseteq L_{\mu}(\A) \right\}$.
\end{definition}

\begin{remark}
As this definition gives a set of forbidden finite words, we clearly see that $\mu$-limit sets are subshifts.
\end{remark}

\begin{definition}[Set of predecessors]
We define the set of predecessors at time $n$ of a finite word $u$ for a CA $\A=(Q,r,\delta)$ as $P^n_{\A}(u)=\left\{v\in Q^{|u|+2rn}: \A^n([v]_{-rn})\subseteq [u]_0\right\}$.
\end{definition}

\begin{remark}
 As we consider the uniform Bernoulli measure $\mu$, $\frac{|P^n_{\A}(u)|}{|Q|^{|u|+2rn}}\to 0$ $\Leftrightarrow$ $u\notin L_{\mu}(\A)$.
\end{remark}
\begin{remark}
 The set of normal configurations has measure $1$ in $Q^{\ZZ}$. Which means that a configuration that is randomly generated according to measure $\mu$ is a normal configuration.

\end{remark}

\begin{lemma}
\label{lem:normal}
Given a CA $\A$ and a finite word $u$, with $\mu$ the uniform Bernoulli measure, for any normal configuration $c$:\\
$u\in \Lambda_{\mu}(\A)$ $\Leftrightarrow$  $d_{\A^n(c)}(u)\nrightarrow 0$ when $n\to +\infty$.
\end{lemma}

\begin{proof}
Let $n\in\NN$, $r$ be the radius of $\A$ and $\mu$ the uniform measure. We prove here that:
$d_{\A^n(c)}(u)=\A^n\mu(u)=\frac{|P_{\A}^n(u)|}{|Q|^{|u|+2rn}}$.\\
The second part of the equality is obtained by definition of $\A^n\mu(u)$. We focus on the first part.
Since any occurence of $u$ in $\A^n(c)$ corresponds to an occurence of a predecessor of $u$ in $c$ :
\begin{displaymath}
d_{\A^n(c)}(u)=\limsup_{k\to +\infty} \frac{|\A^n(c)_{[-k,k]}|_u}{2k+1-|u|}=\limsup_{k\to +\infty} \sum_{v\in P^n_{\A}(u)}\frac{|c_{[-k-rn,k+rn]}|_v}{2k+2rn+1-(|u|+2rn)}.
\end{displaymath}
And as $c$ is normal, for any $v\in  P_{\A}^n(u)$ :
$
|c_{[-k-rn,k+rn]}|_v\sim_{k\to +\infty}\frac{2k+1}{|Q|^{|u|+2rn}}.
$\\
Then:
\begin{displaymath}
d_{\A^n(c)}(u)= \sum_{v\in P^n_{\A}(u)}\limsup_{k\to +\infty}\left(\frac{1}{2k+1-|u|}\frac{2k+1}{|Q|^{|u|+2rn}}\right)=\sum_{v\in P^n_{\A}(u)}\frac{1}{|Q|^{|u|+2rn}}=\frac{|P^n_{\A}(u)|}{|Q|^{|u|+2rn}}.
\end{displaymath}
\end{proof}

\begin{prop}\label{prop:pasteTwoWords}
Let $u\in L_{\mu}(\A)$, there exists a word $w$ such that $uwu\in L_{\mu}(\A)$.
\end{prop}
\begin{proof}
Let $u\in L_{\mu}(\A)$, there exists $\alpha>0$ and an increasing sequence $(n_i)_{i\in\NN}$ such that $\A^{n_i}\mu([u])>\alpha$. Thus, for a normal configuration $c$, one has $d_{\A^{n_i}(c)}(u)>\alpha$ for all $i\in\NN$. Let $l\in\NN$ and $\epsilon>0$ such that  $\frac{2|u|}{2|u|+l}<\alpha-\epsilon$, we define  
\begin{eqnarray*}
W_1&=&\left\{w\in Q_{A}^\ast : u \textrm{ is not a subword of } w \textrm{ and } |w|\leq l\right\} \textrm{ and}\\
 W_2&=&\left\{w\in Q_{A}^\ast : u\textrm{ is not a subword of }w \textrm{ and } |w|> l\right\}.
\end{eqnarray*}

Consider $uW_ku=\{uwu : w\in W_i\}$ for $k\in\{1,2\}$, one has 
$$d_{\A^{n_i}(c)}(uW_2u)  = \limsup_{n\to\infty}\frac{|\A^{n_i}(c)_{[-n,n]}|_{uW_2u}}{2n+1}\leq\frac{2|u|}{2|u|+l}$$
since a word of $uW_2u$ can appear at most $2|u|$ times for each pattern of length $2|u|+l$ of $\A^{n_i}(c)$. Moreover $d_{\A^{n_i}(c)}(uW_1u) +d_{\A^{n_i}(c)}(uW_2u)\geq d_{\A^{n_i}(c)}(u)$ so $d_{\A^{n_i}(c)}(uW_1u)\geq \epsilon$.

Since $W_1$ is finite, there exists a word $w\in W_1$ such that $d_{\A^{n_i}(c)}(uwu)\geq \epsilon$ for an infinity of $i\in\NN$. Thus $uwu\in L_{\mu}(\A)$. 
\end{proof}

\begin{example}[]
We consider here the ``max'' automaton $\A_M$. The alphabet contains only two states $0$ and $1$. The radius is $1$. When the rule applies to three $0$ (no $1$), it produces a $0$. In any other case, it produces a $1$. 

 The probability to have a $0$ at time $t$ is the probability to have $0^{2t+1}$ on the initial configuration. Which tends to $0$ when $t\to \infty$ for the uniform Bernoulli measure. So, $0$ does not appear in the $\mu$-limit set. And finally $\Lambda_{\mu}(\A_M)=\{^{\infty}1^{\infty}\}$.

  And this example gives a difference between subshifts that can be realised as limit set ($\Lambda(\A)=\bigcap_{i\in\NN}\A^i(Q^{\ZZ})$) and subshifts that can be realised as $\mu$-limit set. Effectively, $\Lambda(\A_M)=(^{\infty}10^*1^{\infty})\bigcup (^{\infty}0^{\infty})\bigcup (^{\infty}01^{\infty})\bigcup (^{\infty}10^{\infty})$, but  if we apply proposition \ref{prop:pasteTwoWords} with the word $01$, we conclude that $\Lambda(\A_M)$ cannot be a $\mu$ limit set.

\end{example}

\section{Counters}\label{sec:counters}

In this section and the following one, we describe an automaton $\A_S$, which, on normal configurations, produces finite segments of size growing with time. In these segments, we will  make computations described in section  \ref{sec:infinite_unions}.

Before starting the computation, the automaton $\A_S$ has a transitory regime which erases the random configuration and generate segments between $\#$ where the computation is done. To do that, we have a special state $\ast$, that can only appear in the initial configuration, and which generates two counters. Between two counters, the states are initialized and when two counters intersect, they compare their respective age. If they do not have the same age, the younger deletes the older one; if they have the same age, they disappear and we put the state $\#$ in order to start the computation. The notion of counters was introduced in~\cite{DPST} to produce equicontinuous points according to arbitrary curves. 

We recall some ideas which allow to construct such automaton:
\begin{itemize}
\item no transition rule produces the state $\ast$;
\item $\ast$ produces two couples of signals, one toward the left and another one toward the right; 
\item a couple of signal (called \emph{counter}) is formed by an \emph{inner} signal and an \emph{outer} signal, which is faster. Their collisions are handled in the following way:
\begin{itemize}
\item nothing other than an outer signal can go through another
outer signal;
\item when two outer signals collide they move through each
other and comparison signals are generated;
\item on each side, a signal moves at maximal speed towards the
inner border of the counter, bounces on it and
goes back to the point of collision;
\item the first signal to come back is the one from the youngest
counter and it then moves back to the outer side of the oldest
counter and deletes it;
\item the comparison signal from the older counter that arrives
afterwards  is deleted and will not delete the younger counter's
outer border;

\end{itemize}
\item between a left counter and a right counter, the configuration is initialized;
\item if two counters that have the same age meet, they disappear and produce the state $\#_S$ which start the computation described in section~\ref{sec:infinite_unions}
\item the state $\#_S$ becomes $\#$ which delimitates segments, this state can disappear if two adjacent segments decide to merge as described in section~\ref{sec:segments}, or if a counter (necessarily younger) encounters it. 
\end{itemize}

\begin{figure}
\begin{center}
\includegraphics[scale=1.5]{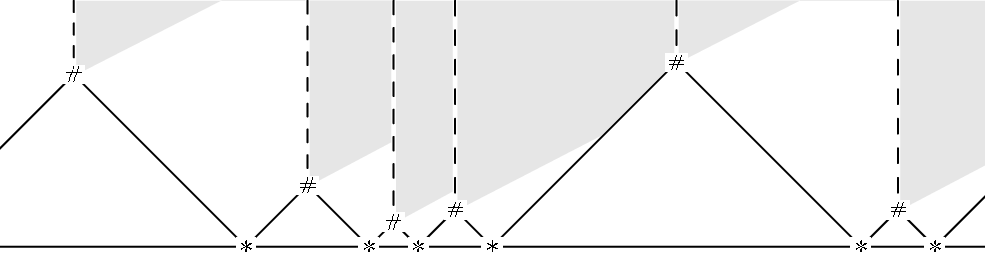}
\caption{When two counters launched by a $\ast$ meet, a $\#$ is produced and a computation is launched on the right. The computation area extends until it meets the inner signal of a counter or another $\#$.}\label{fig:counters}

\end{center}
\end{figure}

The initialization of a configuration is illustrated in figure \ref{fig:counters}. The gray areas of computation begin on the left of a $\#$ produced by the meeting of two counters generated by a $\ast$.

\begin{lemma}\label{lem:timeCounters}
There exists a constant $K_c$ such that if two $\#$ are distant of $k$, they appeared before time $k\times K_c$. 
\end{lemma}
\begin{proof}
Consider two states $\#$ in the space time diagram separated by $k$ cells. If $\#$ is not in the initial configuration, the only way to appear is to result from the collision of two counters coming from the left and from the right. Thus, in the initial configuration, it is necessary to have the state $\ast$ between the two $\#$ to create the two $\#$. This operation take at most  $k\times K_c$ where $K_c$ is the speed of an inner signal. 
\end{proof}

\section{Merging segments}
\label{sec:segments}
We saw in Section \ref{sec:counters}, how a special state $\ast$ on the initial configuration gave birth to counters protecting everything inside them until they meet some other counter born the same way. In this section, we will describe the evolution of the automaton $\A_S$ after this time of initialization. 
 When two counters of the same age meet, they disappear and  a $\#$ is produced.

\begin{definition}[Segment]
	A \emph{segment} $u$ is a subword of a configuration delimited by two $\#$ and containing  no $\#$ inside. So, $u\in \#\left(Q\setminus\{\#\}\right)^*\#$. The \emph{size} of a segment is the number of cells between both $\#$.
\end{definition}

There will be computations made inside segments, but we will describe it later. Thus, in a segment, there is a layer left for computations that remain inside the both $\#$, and a ``merging layer'' that will contain signals necessary to the behavior with other segments. Every signal presented in this section will travel on this merging layer. The idea is the following: at some times, two neighbor segments will decide to merge together to form one single segment whose size will be the sum of both sizes plus one. And we will assure that each segment will eventually merge, so that no segment of finite size can still be in the $\mu$-limit set of $\A_S$.\\

 When a $\#$ is produced in automaton $\A_S$, it sends two signals, on its right and  on its left to detect the first $\#$ on each side. If the signal catches the inside of a counter still in activity before reaching a $\#$,  it waits until the counter produces a $\#$. Then  both $\#$ have recognised each other and the segment between them is ``conscious''. It launches a computation inside it, and waits until it is achieved. We will assure later that this computation ends. When this is done, it will alternatively send signals to its left and to its right in order to propose successively to each neighbor to merge. \\

 For this purpose, it computes and stores the length $n$ of the segment as a binary representation. Then the segment puts a $L$ mark on its left $\#$, and waits for $n^2$ timesteps. If, during this time, the left side neighbor has not put a $R$ mark on the common $\#$, our segment erases the $L$ mark, a signal is sent on the other side, and it puts a $R$ mark on its right $\#$. It waits once again $n^2$ timesteps before erasing the $R$, sending a signal to its left, and starting over. The whole cycle takes $2(n^2+n)$ timesteps as we consider a signal at speed $1$ crossing a segment of size $n$. We request the signal to stay $n^2$ timesteps because as $(n+1)^2>n^2+2n$, if two segments do not have the same size, their signals eventually meet during a cycle of the smallest one. So, the only case in which two neighbor segments that try to merge do not merge, is when they have same size and are correctly synchronized. Computing and storing $n$, and waiting $n^2$ can be done with a space $\log(n)$.\\

 This process ends when at the same time, both a $L$ and a $R$ mark are written on a $\#$. When this happens, the two segments agree to merge together and they do it: the $\#$ between them is erased, and the whole activity begins again, starting with the computation inside the new segment.\\

\begin{figure}

\begin{center}
\includegraphics[scale=1.8]{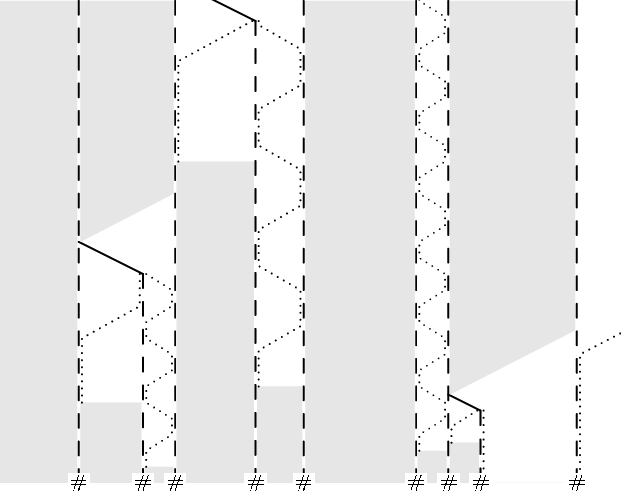}
\caption{A $\#$ stays until two segments merge. Computation happens in gray areas, and at its end a signal (...) is sent and stays on the left of the segment, then goes to the right, stays and comes back. This cycle continues until a neighbor's signal is on the same $\#$ at the same time. Then the $\#$ is deleted and another computation is started on the left.}\label{fig:merging}

\end{center}
\end{figure}

 The general behavior of the segments among themselves is illustrated in figure~\ref{fig:merging}.
We prove the following claims for automaton $\A_S$.

\begin{claim}
\label{lem:diffsize1}
For any two words $u,v\in Q^*$, with $|u|\neq|v|$, if the word $w=\# u\# v\#$ appears at time $t$ in a space-time diagram of $\A_S$, one of the $3$ $\#$ of $w$ has disappeared at time $t+|Q|^{|w|}+2(|w|^2+|w|)$.  
\end{claim}

\proof
If the word $w$ exists at time $t$ on a space time diagram, at time $t+|Q|^{|u|}$ (respectively $t+|Q|^{|v|}$) at most, the computation is achieved in $u$ (resp. $v$). We suppose here that no $\#$ in $w$ has disappeared, which means, $u$ and $v$  do not merge with any other segment outside $w$. So at time $t+|Q|^{|w|}$ both segments try to merge with another one. Assume $|v|> |u|$ for example, the other case is totally symmetric. Then, as $|v|^2>|u|^2+2|u|$, before the end of the cycle of $v$, they have put their mark simultaneously on their common $\#$ for one timestep at least. And consequently, they have merged and one $\#$ has disappeared at time $t+|Q|^{|w|}+2(|w|^2+|w|)$.  
\qed

\begin{claim}
\label{lem:mergetime}
If two segments of size less than $k\in\NN$  merge together, they do it at most $|Q|^{2k}+2((2k)^2+2k)$ timesteps after being formed.
\end{claim}

\proof
If they don't have the same size, lemma \ref{lem:diffsize1} let us conclude. If they have the same size, their computations are achieved after $|Q|^{2k}$. And as their merging cycle takes the same time for both, if they do not merge during the first cycle, they will never merge. So if they merge, they do it before $|Q|^{2k}+2((2k)^2+2k)$.
\qed

\begin{claim}
\label{lem:diffsize2}
For any two words $u,v\in Q^*$, with $|u|\neq|v|$, the word $w=\# u\# v\#$ does not appear in $\Lambda_{\mu}(\A_S)$.
\end{claim}

\proof
We use the constant $K_c$ from lemma \ref{lem:timeCounters}. Denote \\ $T=|w|\times K_c+|w|\left(|Q|^{|w|}+2(|w|^2+|w|)\right)$. We prove that for $t> T$, $P^t_{\A_S}(w)=\emptyset$. \\
 If the two $\#$ encircling $w$ never disappear, the dynamic inside $w$ is not affected by the exterior. Through time, some other $\#$ possibly appeared and disappeared between them. But after time at most $|w|\times K_c$, they have all appeared. Since then, they will only disappear. There are less  than $|w|-1$ excedentary $\#$ that have to disappear. Considering lemma \ref{lem:mergetime}, one disappears at least every $|Q|^{|w|}+2(|w|^2+|w|)$ timesteps. After that, the two segments of $w$ are formed, and with lemma \ref{lem:diffsize1}, one of the $\#$ of $w$ disappear before $|Q|^{|w|}+2(|w|^2+|w|)$ new timesteps. Finally, at time $T$, one of the $\#$ of $w$ has disappeared and $P^t_{\A_S}(w)=\emptyset$.  
\qed

\begin{prop}
\label{prop:nosharp}
There is no $\#$ in the $\mu$-limit set of $\A_S$.

\end{prop}

\proof

Assume that $\#\in L_{\mu}(\A)$, by Proposition~\ref{prop:pasteTwoWords}, there exits $u\in Q^{\ast}$ such that $\# u\#\in L_{\mu}(\A)$, we can assume that $u$ does not contain $\#$. Let $k=|u|$, by Lemma~\ref{lem:timeCounters}, the $\#$ encircling $u$ appeared before time $k \times K_c$. Denote $W=\{\#v\# : v\in(Q\setminus\{\#\})^k\}$ and $X_n=\{x\in Q^{\ZZ} : \A^k(x)_{[0,k+1]}\in W \textrm{ for all } k\in[k\times K_c,n]\}$. Since $\# u\#\in L_{\mu}(\A)$, there exists $\alpha>0$ such that $\mu(X_n)>\alpha$ for an infinity of $n\in\NN$. Moreover, as $X_{n+1}\subset X_n$, we can conclude that $\mu(X_{\infty})>\alpha$ where $X_{\infty}=\cap_{n\in\infty}X_N$. 

As $\mu$ is Bernoulli, we have $\mu(Y)>0$ where $Y=[\ast(Q\smallsetminus \ast)^{2k}\ast(Q\smallsetminus \ast)^{2k}\ast]_0$; moreover there exist $k_1\geq 0$ and $k_2\geq k_1+4k+1$ such that $\mu(Z)>0$ where $Z=X_{\infty}\cap\sigma^{-k_1}(Y)\cap\sigma^{-k_2}(X_{\infty})$. For all $n\geq k\times K_c$ one has $F^n(Z)_{[0,k+1]}\subset W$, $F^n(Z)_{[k_2,k_2+k+1]}\subset W$      and $F^n(Z)_{[k_1+k,k_1+3k]}\subset F^n(Y)_{[k_1+k,k_1+3k]}$ does not contain $\#$. 

We deduce that there exists a word $w\in Q$ of length $k_2-k-1$ such that $w_{[k_1+k,k_1+3k]}$ does not contain $\#$ and $\#u\#w\#u\#\in L_{\mu}(\A)$. However, in $\#u\#w\#u\#$ we can find two segments $\#u_1\#u_2\#$ which have different length. By Claim~\ref{lem:diffsize2} we obtain a contradiction. Thus, there is no $\#$ in the $\mu$-limit set of $\A_S$.
\qed

Finally, we prove a lemma that will be useful later.

\begin{claim}
\label{cla:outofsegments}
The density of cells outside segments generated by counters born in the initial configuration tends to $0$.
\end{claim}
\proof
The proof is clear since such a cell needs predecessors without states $\ast$ on each side in the initial configuration.
\qed

\begin{lemma}
\label{lem:densityinside}
 Let $u\in Q^*$. If $\forall k,\forall l\geq k$, for any segment $v\in Q^l$, $d_v(u)\leq \alpha_{k}$ with $\alpha_k\to 0$ when $k\to \infty$, then $u\notin L_{\mu}(\A_S)$.\\
 Conversely, if $\forall k,\forall l\geq k$, for any segment $v\in Q^l$, $d_v(u)\geq \alpha_{k}$ with $\alpha_k\nrightarrow 0$ when $k\to \infty$, then $u\in L_{\mu}(\A_S)$.
\end{lemma}
\proof
Let's consider a normal configuration $c$. For any $k\in \NN$, we denote \begin{displaymath}
d_k^t=\sum_{v\in \#(Q_{\A})^l\#,\ l\leq k}l\times d_{\A_S^t(c)}(v)
\end{displaymath} the density of cells in segments of size less than $k$ in the image at time $t$ of $c$. Due to proposition \ref{prop:nosharp}, $d_k^t\to 0$ when $t\to \infty$. And due to claim \ref{cla:outofsegments}, the density $a^t$ of cells outside wellformed  segments tends to $0$ when $t\to \infty$.

 Suppose $\forall k,\forall l\geq k,\forall v\in Q^l$ segment, $d_v(u)\leq \alpha_{k}$ and $\alpha_k\to 0$. Any occurence of $u$ is either in a segment of size less than $k$, either in a segment of size greater than $k$, or out of segments. Finally, at a given time $t$, $d_{\A_S^t(c)}(u)\leq d_k^t + \alpha_k+ a^t$.

 As this equation holds for any $k$, finally, when $t\to \infty$, $d_{\A_S^t(c)}(u)$ has a limit which is $0$. This concludes the proof of the first part of the lemma with lemma \ref{lem:normal}. \\

In the other side, suppose $\forall k,\forall l\geq k,\forall v\in Q^l$ segment, $d_v(u)\geq \alpha_{k}$ and $\alpha_k\nrightarrow 0$.   Therefore,  $d_{\A_S^t(c)}(u)\geq (1-d_{k}^t-a^t)\alpha_k$ which does not tend to $0$ when $t\to \infty$ and $k\to \infty$. Thus, $u\in L_{\mu}(\A_S)$.
\qed

\section{Infinite Unions} 
\label{sec:infinite_unions}
In this section we will see how to create a CA whose $\mu$-limit set is the closure of the infinite union of a recursively enumerable family of particular subshifts.

\begin{definition}[Generable Subshift]
	We say that a Turing machine $M$ \emph{generates} a subshift $\Si\subseteq Q^\ZZ$ if $M$ computes a generic configuration of $\Si$ in the following sense:
	\begin{itemize}
		\item the tape alphabet of $M$ contains $Q$;
		\item on an empty tape, $M$ writes the right half of a configuration $c\in\Si$ such that $\limsup_{n\to\infty}\frac{|c_{[0,n]}|_u}{n+1}>0$ if and only if $u\in\Lang(\Si)$; $c$ is called a \emph{generic configuration};
		\item after a symbol of $Q$ has been written on the tape, it is never changed.
	\end{itemize}
\end{definition}

\begin{theorem}
	Given a recursively enumerable family $(\Si_i)_{i\in\NN}$ of generable subshifts, that is to say that there exists a Turing machine that enumerates a set of machines $(M_i)_i\in\NN$ such that $M_i$ generates the subshift $\Si_i$, there exists a cellular automaton $\A$ whose $\mu$-limit set is exactly the subshift $\overline{\bigcup_{i\in\NN}\Si_i}$.
\end{theorem}

\proof 
	Let us consider a recursively enumerable family $(\Si_i)_{i\in\NN}$ of generable subshifts, let us denote by $M$ the Turing machine that enumerates the machines $(M_i)_{i\in\NN}$ such that $M_i$ generates the subshift $\Si_i$.
	
	We now describe the behavior of such a cellular automaton $\A$. $\A$ will work as the automaton $\A_S$ described in Section \ref{sec:segments}: starting from a normal configuration, it will generate ``counter signals'' that will produce finite segments on the configuration (separated by a $\#$ symbol). We now describe the computation performed by each finite segment during the evolution of the cellular automaton.
	
	The first thing a segment does is compute its length $n$ and store it as a binary number. By incrementing a binary counter moving across the segment, this is easily done in space $\log(n)$. Once this is done, the segment can simulate Turing machines on its first $\log(n)$ cells (it is important to limit the computational space so that the computation states become negligible and disappear from the $\mu$-limit set).
	
	On the initial $\log(n)$ cells of the segment the machine $M$ is simulated to produce the descriptions of the first $k$ machines $(M_i)_{i<k}$, with $k$ as big as possible for $M$ computing on a tape of length $log(n)$. And we also request that  $k\leq \log(\log(n))$. $k$ may be 0 for short segments, but we know that as the segments grow larger, $k$ will grow too.
	
	The space of size $\log(n)$ is further divided into $k$ fragments of size $\log(n)/k$. On the $i$-th fragment, the corresponding machine $M_i$ is simulated to produce the word $w_i$ beeing the begining of the generic configuration corresponding to the subshift $\Si_i$. The word $w_i$ might be much smaller than $\log(n)/k$ depending on the space needed by the machine $M_i$ to compute, but again we know that as segments grow larger, larger words will be computed.
	
	After the $k$ different $w_i$ have been computed, the initial segment of length $n$ is split into $\sqrt n$ fragments of length $\sqrt n$. Each of these fragments is filled with copies of one of the $w_i$ in the following manner: one out of two is filled with $w_1$, one out of four (i.e. one out of two among the remaining fragments) is filled with $w_2$, one out of eight is filled with $w_3$ and so on. The remaining segments (if $k$ is very small, we might run out of $w_i$ before filling all the fragments) are filled with $w_k$.
	Fragments are separated by a symbol $\$_1\notin Q$ and the copies of words $w_i$ inside a given fragment are separated by a symbol $\$_2\notin Q$.

	\begin{remark}
		The previous construction can be done using only $\log(n)$ cells of computation at each step (cells that are not active and that only contain a symbol from $Q\cup\{\$_1, \$_2\}$ are not counted). To fill the fragments of size $\sqrt n$ we only need to compute the binary expression of $\sqrt n$ and then advance through the segment while filling the fragment with the appropriate $w_i$ while decreasing a counter to measure $\sqrt n$ cells. The important data (the words $w_i$ and different counters) are moved through the segment so that they are always present near the location to be filled. Thus the head of the Turing machine $M$ carries only $\log(n)$ cells used to store the $w_i$ and to its computation. No mark of the computation remains in the other cells, even those already visited and rewritten.
	\end{remark}
	
	When all the fragments of the segment have been filled with the $w_i$, the segment can erase all the remaining computation data and start the process of merging with its neighbors as described in Section \ref{sec:segments}.
	
	When two segments merge, the whole computation is restarted but this time with a larger space. The segments are not erased immediately after a merge, but rather the new data overwrites the previous as the $\sqrt n$ fragments are filled.

 We will prove that $L_{\mu}(\A)=\bigcup_{i\in\NN}\Lang(\Si_i)$.

\begin{claim}
\label{cla:computationstates}
The states used for computation, signals inside segments, writing fragments, $\$_1$ or $\$_2$ do not appear in $\Lambda_{\mu}(\A)$.
\end{claim}
\proof
 Here we use the lemma \ref{lem:densityinside} for each of these states.

 We use the $\log(k)$ initial cells of a segment of size $k$ to do the computation, so  the density of these cells is $\log(k)/k$, and the property is proved.
 The head of the Turing machine $M$ carries at most $\log(k)$ cells for its computation or writing, thus the same argument works. The signals for the merging process are in a finite number in a segment, therefore their density in a segment tends to $0$ too. The density of $\$_1$ is $\sqrt{k}/k$, and the lemma applies once again. 

 For the density of $\$_2$, let $\lambda>0$, $\exists k_0>0$ such that  the word $w_i$ produced in  a segment of size $k>k_0$ is such that $|w_i|>\lambda$ for any $i\leq\lambda$. So, for $k>k_0$, the density of $\$_2$ in a segment of size $k$ is less than $1/\lambda$ in fragments of $S_i, i\leq\lambda$ and less than $1$ in the other fragments that have themselves a density lower than $1/2^{\lambda}$. And thus, the density of $\$_2$ is lower than $\frac{1}{\lambda}+\frac{1}{2^{\lambda}}$ in segments of size $k>k_0$. Finally the density of $\$_2$ tends to $0$ when  $k\to \infty$. And the claim is proved.
\qed

\begin{claim}
\label{cla:worddensity}
For any subshift $\Si_i,\ i\in \NN$, any word $u\in \Lang(\Si_i)$ and any family of segments $(v_k)_k$ of size $|v_k|=k$, $d_{v_k}(u)$ does not tend to $0$ when $k\to \infty$.
\end{claim}
\proof
As $u\in \Lang(\Si_i)$, its density $\alpha(u)$ in the generic configuration computed by $M_i$ is positive. So, there exists $l_i\in \NN$ such that any subword of this configuration contains $u$ with density at least $\alpha(u)/2$. Let $k_0$ such that in any segment of size $k>k_0$, the word $w_i$ computed has length $|w_i|>l_i$. 

 For any segment $v_k$ of size $k>k_0$, there are $\log(k)$ cells occupied for computation, less than $\sqrt(k)$ cells containing a $\$_1$ and $\frac{1}{2^{i+1}}$ among the remaining cells attributed to the copies of $w_i$. Among these copies, a proportion $\frac{l_i-1}{l_i}$ of the cells contain $\$_2$. $\log(k)$ additional cells can be dedicated to the head of the Turing machine $M$ writing in the segment and a finite number $K$ of cells can contain signals for the merging process. Finally, 
\begin{displaymath}
d_{v_k}(u)\geq \left(\left(\frac{k-\log(k)-\sqrt(k)}{2^{i+1}}\right)\frac{l_i-1}{l_i}-\log(k)-K\right)\frac{1}{k}.
\end{displaymath}
Which does not tend to $0$ when $k\to \infty$. 
\qed

\begin{claim}
\label{cla:wordsinmulimit}
For any subshift $\Si_i,\ i\in \NN$ and any word $u\in \Lang(\Si_i)$, $u\in L_{\mu}(\A)$.
\end{claim}
\proof
We clearly get the result by combining claim \ref{cla:worddensity} and lemma \ref{lem:densityinside}. 
\qed

\vspace{1cm}

Finally, the theorem is proven:
\begin{itemize}
\item the proposition \ref{prop:nosharp} and the claim \ref{cla:computationstates} assure that every state used for computation does not appear in  $\Lambda_{\mu}(\A)$, which means $L_{\mu}(\A)\subseteq \bigcup_{i\in\NN}\Lang(\Si_i)$,
\item the claim \ref{cla:wordsinmulimit} assures that $\bigcup_{i\in\NN}\Lang(\Si_i)\subseteq L_{\mu}(\A)$.
\qed 
\end{itemize}

The next proposition gives some examples of generable subshifts.
\begin{proposition}
The following subshifts are generable: 
\begin{itemize}
\item transitive sofic subshifts,
\item substitutive subshift associated to a primitive substitution.
\end{itemize}
\end{proposition}

\proof
As a transitive sofic subshift $\Si$ is given by the strongly connected automaton recognizing its language. For example, we can write successively every cycle of size $k$ for $k$ from $1$ to $\infty$. In this case we obtain a configuration where the density of all the words of the language of $\Si$ is positive. 

For a primitive substitution $s$, it is easy to generate the fix point configuration denoted $c_{[0;\infty]}$ whose all prefixes are given by $s^k(a)$ for all $k\in\NN$ where $a\in Q$. It is well know that all words of the substitutive subshift associated appears with a positive density in $c_{[0;\infty]}$~\cite{Fog05}.  
\qed

\section{Conclusion and perspectives}

In this paper, we prove that a large class of subshifts can be realized as
$\mu$-limit sets of cellular automata. In particular, it is possible to
obtain all transitive sofic subshifts, this is a profound difference with the
topological case since the even shift cannot be realized as the limit set
of one cellular automaton. This construction allows to control the
iterations of a random configuration in view to obtain an auto-organized
behavior. The construction can be adapted at least in two ways:
\begin{itemize}
\item to obtain the same result for a large class of measure ($\sigma$-ergodic
measure of full support) modulo some technical changes
\item to obtain a subshift without any word of low complexity (as suggested by
V. Poupet).
\end{itemize}

Of course the main open question is in the reciprocal of the theorem,
that is to say to characterize  subshifts that can possibly be realized as
$\mu$-limit sets.
\section*{Acknowledgments}
We are deeply grateful to Victor Poupet and Guillaume Theyssier for their ideas, and constant support to the writing of this article.

\def\ocirc#1{\ifmmode\setbox0=\hbox{$#1$}\dimen0=\ht0 \advance\dimen0
  by1pt\rlap{\hbox to\wd0{\hss\raise\dimen0
  \hbox{\hskip.2em$\scriptscriptstyle\circ$}\hss}}#1\else {\accent"17 #1}\fi}

\end{document}